\documentclass[runningheads]{llncs}

\usepackage{graphicx}

\usepackage[english]{babel}
\usepackage{amsfonts}
\usepackage{inputenc}
\usepackage[ruled]{algorithm}
\usepackage{color}
\usepackage{algpseudocode}
\usepackage{booktabs,multicol,multirow}
\usepackage{fancyhdr}
\usepackage{epic}
\usepackage{xspace}
\usepackage{amssymb,amsmath,verbatim}
\usepackage{multirow}
\usepackage{multicol}
\usepackage{adjustbox}
\usepackage{threeparttable}
\usepackage{float}

\usepackage{amsthm}
\usepackage{chngcntr}
\usepackage{makecell}

\usepackage{lipsum}
\newtheorem{thm}{Theorem}

\newtheorem{rem}[thm]{Remark}

\newtheorem{defn}[thm]{Definition}

\providecommand{\keywords}[1]{\textbf{\textit{Key words: }} #1}

\def\ZZ{\mathbb{Z}}

\def\RR{\mathbb{R}}

\def\cal{\mathcal}
\def\bf{\mathbf}

\def\TrapGen{\mathsf{TrapGen}}

\def\SK{\mathsf{SK}}
\def\PP{\mathsf{PP}}
\def\MSK{\mathsf{MSK}}
\def\SKS{\mathsf{sk_s}}

\def\CT{\mathsf{CT}}
\def\SampleLeft{\mathsf{SampleLeft}}
\def\SampleRight{\mathsf{SampleRight}}
\def\SampleBasisLeft{\mathsf{SampleBasisLeft}}
\def\SampleBasisRight{\mathsf{SampleBasisRight}}
\def\td{\mathsf{td}}

\def\Setup{\mathsf{Setup}}

\def\Enc{\mathsf{Enc}}
\def\Extract{\mathsf{Extract}}
\def\Ext{\mathsf{Ext}}
\def\Dec{\mathsf{Dec}}
\def\Td{\mathsf{Td}}
\def\Test{\mathsf{Test}}
\def\Pr{\mathrm{Pr}}
\def\Adv{\mathsf{Adv}}
\def\OW{\textsf{OW-ID-CPA}}
\def\IND{\textsf{IND-ID-CPA}}

\def\a{\alpha}
\def\u{\bf{u}}
\def\v{\bf{v}}

\def\e{\bf{e}}
\def\d{\bf{d}}
\def\m{\bf{m}}
\def\x{\bf{x}}

\def\L{\Lambda}
\def\Lp{\Lambda^{\perp}}
\def\b{\bf{b}}
\def\s{\bf{s}}
\def\c{\bf{c}}

\def\ID{\mathsf{ID}}
\def\id{\mathsf{id}}

\def\W{\cal{W}}

\begin{document}
	
	\title{Lattice-based IBE with Equality Test Supporting Flexible Authorization in the Standard Model}
	\titlerunning{Lattice-based IBEET-FA in the Standard Model}
	
	\authorrunning{G. L. D. Nguyen, W. Susilo, D. H. Duong, H. Q. Le, F. Guo}
	\author{Giang L. D. Nguyen\inst{3}, Willy Susilo\inst{1}, Dung Hoang Duong\inst{1}, Huy Quoc Le\inst{1,2}\and Fuchun Guo\inst{1}}
	\institute{
		Institute of Cybersecurity and Cryptology, School of Computing and Information Technology, University of Wollongong\\
		Northfields Avenue, Wollongong NSW 2522, Australia\\
		\email{\{hduong,wsusilo,fuchun\}@uow.edu.au},
		\email{qhl576@uowmail.edu.au}
		\and
		CSIRO Data61, Sydney, NSW, Australia\\
		\and
		Futurify Software Development\\
		17 Street No. 2, Ward 4, District 3, Ho Chi Minh City, Vietnam\\
		\email{ndlgiang.edu@gmail.com}	
	}
	\maketitle
	
	\begin{abstract}
		Identity-based encryption with equality test supporting flexible authorization (IBEET-FA) allows the equality test of underlying messages of two ciphertexts while strengthens privacy protection by allowing users (identities) to control the comparison of their ciphertexts with others. IBEET by itself has a wide range of useful applicable domain such as keyword search on encrypted data, database partitioning for efficient encrypted data management, personal health record systems, and spam filtering in encrypted email systems. The flexible authorization will enhance privacy protection of IBEET. In this paper, we propose an efficient construction of IBEET-FA system based on the hardness of learning with error (LWE) problem. Our security proof holds in the standard model.
	\end{abstract}
	
	
	\section{Introduction}
	The concept of identity-based encryption with equality test supporting flexible authorization (IBEET-FA) is the combination of public key encryption with equality test supporting flexible authorization (PKEET-FA) and identity-based encryption (IBE). IBE is a public key encryption mechanism where a public key can be an arbitrary string such as email address or phone number. The administrator who has knowledge of master secret key then can generate the corresponding private key by using Private-Key Generator (PKG). IBE by itself can simplify the certificate management of public key encryption. Identity-based encryption scheme with equality test (IBEET) was first introduced by Ma~\cite{IBEET-FA-Ma16}. It supports the equality test of underlying message of two ciphertext with the help of trapdoors. Hence IBEET has found a wide range of applications especially in smart city applications \cite{WZCH18}, in cloud computing such as partition of encrypted emails ~\cite{Ma16}, malware detection and verifiability of encrypted data \cite{AAHM18}, and in  wireless body area networks~\cite{IBEET-RSA}. The flexible authorization (FA) notion was first introduced by Ma et al.  ~\cite{PKEET-FA-Maetal-16}. It enhances the privacy protection of data owners by allowing the users (identities) more choices  in controlling of what they want to compare, either between identities' ciphertexts, between two ciphertexts or between a specific ciphertext and other identity's ciphertexts.
	
	One particular application scenario of using IBEET in smart city is in smart healthcare. Assume that, each family in the smart city has a smart device for monitoring their daily health, such as sleeping time, heart beat rate, how many steps a person walk a day, and so on. Their data is encrypted and stored in the cloud sever managed by the government and the encrypted data can be accessed by the healthcare department from the government. It is a common request that the healthcare officers need to get some residents' health information in order to monitor the community's health. The officers then can perform equality test for specific keywords on encrypted data stored in the cloud server. For this scenario, we can see that Type-1 is needed since  this type of authorization allows the officers to do equality test on all ciphertexts from the users in the community.
	
	Let consider the case  in which some third party, such as a medicine company, can harm the privacy if the system (comprising of smart healthcare devices) is only equipped with the Type-1 authorization.  Suppose that, the medicine company wishes to know the body mass index (BMI) of people in the community for their new kind of weight loss pill.  Obviously, the advertising company can encrypt a specific keyword relating to the BMI (e.g., the key word could be  `` obese" ) using any or all public key(s) of users in the community. Therefore, if the  advertising companies are given the Type-1 method,  the company can perform equality test on all available ciphertexts in the database, which is not good for the residents' privacy.   However, it is a standard privacy policy by law that users can opt-in or opt-out for using of their data.  Then, in this case, the users can choose the Type-2 and/or Type-3 authorization which definitely help to enhance the privacy protection.
	
	Note that in the identity-based encryption (IBE), we have a reduction from the IND security to the OW security (see \cite[Table 1]{GH05}). However, note that in the case of IBEET (hence, IBEET-FA), we do not always have the reduction. The reason for this is that, an OW adversary is able to get the trapdoor according to the target user, while an IND adversary is not allowed to get the target user's trapdoor. Then both IND and OW should be simultaneously required for the IBEET-FA. Moreover, choosing CPA or CCA, CCA2 security depends on the real-life applications to which the IBEET-FA is applied. Recall that, the CPA security ensures that, even if the attacker can get ciphertexts for arbitrary plaintexts that it wants, the  attacker's advantage in  revealing all or part of the secret encryption key is negligible. While the CCA and CCA2 securities guarantee that, even if the attackers can obtain the ciphertexts and the corresponding plaintexts, their advantage in attacks is still negligible. For IBEET and IBEET-FA, it is folkore to consider the security model which involves \IND~(indistinguishability under chosen-identity and under chosen-plaintext attacks) and \OW~(one-wayness under chosen-identity and under chosen-plaintext attacks) into account.\\
		
	\noindent \textbf{Related Works.}	 
	There have been a lot of works involving with PKEET and IBEET such as Lee et al. \cite{Lee2016}, Wu et al. \cite{WZCH18} and Duong et al. \cite{IBEET-Duong19} just to name a few. 	Lee et al. \cite{Lee2016} inherited knowledge from the existing proposals proven secure under the hardness of number theoretic assumptions and in the random oracle (ROM). The authors  introduced a generic construction for PKEET (and also for IBEET) which is secure  in the standard model (SDM) being able to be instantiated in the lattice setting. Their constructions based on two-level hierarchical identity-based, a strong unforgeable one time signature and a cryptographic hash function. Wu et al. \cite{WZCH18}  proposed an efficient
	IBEET scheme secure in the random oracle model (ROM) replying on the Computational Diffie-Hellman (CDH) assumption. Duong et al.~\cite{IBEET-Duong19} has introduced an efficient direct construction of an IBEET based on lattices which is secure in the standard model (SDM). Their method exploits the adaptive identity-based encryption (Full-IBE) scheme which is a post-quantum instantiation based on lattices proposed by ~\cite{ABB10-EuroCrypt}.
		
	Regarding the encryptions that support equality test with FA, there have been a few of works in the literature, e.g., 	Ma et al.  ~\cite{PKEET-FA-Maetal-16}, Duong et al. \cite{PKEET-FA-Duong20}.  As the first PKEET-FA, Ma et al. \cite{PKEET-FA-Maetal-16} enhanced the original PKEET scheme which lacks of a solution for users to control the comparison of their ciphertexts with others' by adding authorization mechanism. The authors also introduced new definitions, security policies, and a PKEET construction supporting four type of authorization proven under the CDH assumption and in the ROM. Recently, Duong et al. \cite{PKEET-FA-Duong20} have studied Ma et al. works and proposed the first quantum-safe PKEET supporting flexible authorization  which are secure in the SDM based on the integer lattice and ideal lattice hard problems.
	
			 	 
\begin{table}[h]
  	\centering
  	 \caption{ A summary of some PKEET  and IBEET with/without flexible authorization.  Here DBDH stands for the decisional
  		Billinear Diffie-Hellman assumption. }
  	\label{tab2}
  	\small\addtolength{\tabcolsep}{0pt}
  	\begin{tabular}{ c | c| c|c|c}
  	\hline
  		\textbf{Literature}& \textbf{PKEET/IBEET}&\textbf{Hard Assumption}&\textbf{Security Model}&\textbf{with FA? }\\
\hline
	Ma et al. \cite{PKEET-FA-Maetal-16}&PKEET&CDH&ROM& $\checkmark$\\
		\hline
	Lee et al. \cite{Lee2016}&PKEET \& IBEET& \makecell{CDH+DBDH \\or LWE}&SDM&$\times$\\
		\hline	
			Wu et al. \cite{WZCH18}&IBEET&CDH&ROM&$\times$\\
\hline	
	Duong et al. \cite{IBEET-Duong19}&IBEET&LWE&SDM&$\times$\\				\hline
Duong et al. \cite{PKEET-FA-Duong20}&PKEET&LWE&SDM&$\checkmark$\\
\hline
\hline
\textbf{This work}&IBEET&LWE&SDM&$\checkmark$\\
	 	\hline
  	\end{tabular} 

\end{table}
	We summarize some related works in Table \ref{tab2} below. According to the best of our knowledge, there have not been any post-quantum secure IBEET that supports a flexible authorization mechanism. Therefore, it is necessary to construct an identity-based encryption with equality test supporting flexible authorization in the standard model and secure in quantum era.
	
	\noindent\textbf{Our contribution.}
	We propose the first concrete construction of an IBEET-FA scheme secure in the standard model based on the hardness assumption of lattices. In particular, our IBEET-FA construction offers the $\IND$ and $\OW$ security. However, we emphasize that our scheme can achieve CCA2 using  the BCHK's transformation~\cite{BCHK07} with the cost of increasing ciphertext size.\\
	
	\begin{table}[pt]	
		\caption{Comparison of our \textsf{IBEET-FA} with other IBEET constructions. Data sizes are in number of field elements. Here $t$ is the length of messages, $\ell$ is the length of identity.}
		
		\centering
		\medskip
		\smallskip
		\small\addtolength{\tabcolsep}{-1pt}
		\begin{tabular}{ c | c | c | c | c }
			\hline
			\textbf{Scheme}&\textbf{Ciphertext}&\textbf{Public key}&\textbf{Master SK}&\textbf{Secret key}\\
			\hline\hline
			Our proposed IBEET-FA& $m^2 + 2t + 6m + \lambda$ &$(\ell+3)mn + nt$ &$2m^2$ &$4m^2$\\
			IBEET of Duong et al. \cite{IBEET-Duong19}& $2t + 4m$ &$(\ell+3)mn + nt$ &$2m^2$ &$4mt$\\
			IBEET of Lee et al. \cite{Lee2016}& $8m + 2t +2mt$ &$(\ell+3)mn + nt$ &$2m^2$ &$2mt$\\
			\hline
		\end{tabular} 
		
		\label{tab3}
	\end{table}
	
	\noindent\textbf{Technical Overview.} 
	Our idea is from the works of Duong et al.~\cite{IBEET-Duong19,PKEET-FA-Duong20}. Recall that, the PKEET-FA of Duong et al.~\cite{PKEET-FA-Duong20} exploits the full-IBE construction by Agrawal et al.~\cite{ABB10-EuroCrypt}. In particular:
	\begin{itemize}
		\item Each user has their own pair of (PK, SK).
		\item In \textsf{Encrypt} method, they choose a vector tag $\b=(b_1,\cdots,b_\ell)\in\{-1,1\}^\ell$ and then compute $F_1:=(A|B + \sum_{i=1}^\ell b_iA_i)$. That means each user can produce many ciphertexts, each of which contains one different vector tag $\b$.
		\item The vector tag $\b$ is then combined with the trapdoor (which is the second part of the secret key SK) for flexible authorization methods.
	\end{itemize}
	
	Although our IBEET-FA work is inspired from \cite{PKEET-FA-Duong20} and \cite{ABB10-EuroCrypt}, we face the following obstacles:
	\begin{itemize}
		\item The whole system only has one pair of (MPK, MSK).
		\item Each user has their own $\ID=(\id_1,\cdots,\id_\ell)\in\{-1,1\}^\ell$, which is used to compute the matrix $F_\ID=(A|B + \sum_{i=1}^\ell\id_iA_i)$ in \textsf{Extract} algorithm to produce their secret key from MSK. Each user can produce multiple ciphertexts. However, those all ciphertexts contain the same \textsf{ID} of the user. That's why the \textsf{ID} cannot be used as the role of vector tag $\b$. 
		\item We need a new tag (which can be a random vector or matrix) to play the role of the vector tag $\b$. This tag will be appended to $F_\ID$ and be stored in ciphertext for flexible authorization later.
	\end{itemize}

	
	Therefore, in order  to support FA with many different types, we need to provide additional tag which is random and specified for each ciphertext. Namely, we sample uniformly at random a matrix $R$ (called  \textit{matrix tag}) and build the matrix $F=(A|B + \sum_{i=1}^\ell\id_iA_i|AR)$ which will be used during  the encryption algorithm. By using the matrix $AR$ like this, we can also construct a reduction from the learning with errors (LWE) problem to the $\IND$ (and $\OW$) security of the proposed IBEET-FA. If we replace $AR$ by some random matrix, the reduction will not work. In summary, we use the matrix $AR$ for two purposes:
	\begin{itemize}
		\item \textbf{P1:} To combine with trapdoor for different Type-$\alpha$ flexible authorisation.
		\item \textbf{P2:} For the reduction from LWE to the $\IND$ and to $\OW$ security to work.
	\end{itemize}
	
The ciphertext in our scheme has this form $\CT = (R,\c_1,\c_2,\c_3,\c_4,\c_5)$, where $R$ is the matrix tag to support flexible authorization, $(\c_1, \c_3)$ is the encryption of message \textbf{m} and $(\c_2, \c_4)$ is the encryption of hash value H(\textbf{m}). Both $R$ and $(\c_2, \c_4)$ are used to support flexible authorization. Here $\c_5=H'(\c_1||\c_2||\c_3||\c_4)$ is used to strengthen the FA procedures by checking  the consistence of ciphertexts before executing equality test. 
		
	We present in Table \ref{tab3} a comparison of our \textsf{IBEET-FA} with the lattice-based IBEETs of  Lee et al. \cite{Lee2016} and Duong et al.  \cite{IBEET-Duong19}. \\

\noindent \textbf{The weakness of the existing security model for IBEET and IBEET-FA.} We remark that the component $\c_5$ in the ciphertext of our IBEET-FA scheme is only used to  check the consistence of ciphertexts before executing equality test. Without $\c_5$ (i.e., assume that $\CT = (R,\c_1,\c_2,\c_3,\c_4)$),  our scheme still enjoys the $\IND$ and $\OW$ security. However, for real-life applications, our construction with $\CT = (R,\c_1,\c_2,\c_3,\c_4)$ may suffer from the following attack. 
Suppose that we use equality test on encrypted data for a research of counting the number of the obese people, each ciphertext corresponds to each person. Then, if an adversary can create malformed ciphertexts from some original ciphertext by keeping $(R, \c_2,\c_4)$ but replacing the pair of $(\c_1,\c_3)$, then the database would be changed in the number of ciphertexts having the same underlying message, which results in the wrong output of the research. Obviously, this also means that the existing security model for IBEET and IBEET-FA is not strong enough in some manners.
	
	\section{Preliminaries}\label{sec:prelim}
	
	\subsection{Identity-based encryption with equality test supporting flexible authorization (IBEET-FA)}
	In this section, we will recall the model of IBEET-FA and its security model.
	\begin{defn}\label{def:PKEET-FA}
		An IBEET-FA scheme consists of the following polynomial-time algorithms:
		\begin{itemize}
			\item $\Setup(\lambda)$: On input a security parameter $\lambda$, it outputs a public parameter $\PP$ and a master secret key $\MSK$.
			\item $\Extract(\PP,\MSK,\ID)$: On input $\PP, \MSK$ and an identity $\ID$, it outputs an identity $\ID$'s secret key $\SK_{\ID}$.
			\item $\Enc(\PP,\ID,\bf{m})$: On input  $\PP$, an identity $\ID$ and a message $\bf{m}$, it outputs a ciphertext $\CT$.
			\item $\Dec(\PP,\SK_{\ID},\CT)$: On input $\PP$, an identity $\ID$'s secret key $\SK_\ID$ and a ciphertext $\CT$, it outputs a message $\bf{m}'$ or $\perp$.
			
			Note that $\PP$ contain the information of the message space $\cal{M}$. It is used in other algorithms implicitly.
		\end{itemize}  
		
		Suppose that an identity $\ID_{i}$ has secret key $\SK_{i}$, whose ciphertext is $\CT_{i}$, and an identity $\ID_{j}$ has secret key $\SK_{j}$, whose ciphertext is $\CT_{j}$ respectively. The Type-$\alpha$ ($\alpha = 1, 2, 3$) authorization for $\ID_{i}$ and $\ID_{j}$ consist of two algorithms:
		\begin{itemize}
			\item[$\bullet$] $\Td_{\alpha}$ ($\alpha = 1, 2, 3$): an algorithm to generate trapdoors for $\ID_{i}$’s ciphertexts that are required to be compared with $\ID_{j}$’s ciphertexts.
			\item[$\bullet$] $\Test_{\alpha}$ ($\alpha = 1, 2, 3$): an algorithm to check and compare whether two identities’ ciphertexts contain the same underlying message or not.
		\end{itemize}
		
		Type-1 Authorization:
		\begin{itemize}
			\item[$\bullet$] $\Td_1(\SK_{\ID})$: On input the secret key $\SK_{\ID}$ of the identity $\ID$, it outputs a trapdoor $\td_{1,\ID}$ for that identity.
			\item[$\bullet$] $\Test_1(\td_{1,\ID_i},\td_{1,\ID_j},\CT_{\ID_i},\CT_{\ID_j})$: On input trapdoors $\td_{1,\ID_i}, \td_{1,\ID_j}$ and ciphertexts $\CT_{\ID_i}, \CT_{\ID_j}$ for identities $\ID_i, \ID_j$ respectively, it outputs $1$ or $0$.
		\end{itemize}
		
		Type-2 Authorization:			
		\begin{itemize}
			\item[$\bullet$] $\Td_2(\SK_{\ID}, \CT_{\ID})$: On input the secret key $\SK_{\ID}$ and ciphertext $\CT_{\ID}$ of the identity $\ID$, it outputs a trapdoor $\td_{2,\ID}$ for $\CT_{\ID}$ of that identity.
			\item[$\bullet$] $\Test_2(\td_{2,\ID_i},\td_{2,\ID_j},\CT_{\ID_i},\CT_{\ID_j})$: On input trapdoors $\td_{2,\ID_i}, \td_{2,\ID_j}$ and ciphertexts $\CT_{\ID_i}, \CT_{\ID_j}$ for identities $\ID_i, \ID_j$ respectively, it outputs $1$ or $0$.
		\end{itemize}
		
		Type-3 Authorization:
		\begin{itemize}			
			\item[$\bullet$] $\Td_3(\SK_{\ID})$: On input the secret key $\SK_{\ID}$ of the identity $\ID$, it outputs a trapdoor $\td_{3,\ID}$ for identity $\ID$.
			\item[$\bullet$] $\Td_3(\SK_{\ID}, \CT_{\ID})$: On input the secret key $\SK_{\ID}$ and ciphertext $\CT_{\ID}$ of the identity $\ID$, it outputs a trapdoor $\td_{3,\ID}$ for $\CT_{\ID}$ of identity $\ID$.
			\item[$\bullet$] $\Test_3(\td_{3,\ID_i},\td_{3,\ID_j},\CT_{\ID_i},\CT_{\ID_j})$: On input trapdoors $\td_{3,\ID_i}, \td_{3,\ID_j}$ and ciphertexts $\CT_{\ID_i}, \CT_{\ID_j}$ for identities $\ID_i, \ID_j$ respectively, it outputs $1$ or $0$.
		\end{itemize}    
	\end{defn}
	
	\noindent\textbf{Correctness.} We say that an IBEET-FA scheme is \textit{correct} if the following conditions hold:
	\begin{description}
		\item[(1)] For any security parameter $\lambda$, any identity $\ID$ and any message $\bf{m}$:
		$$\Pr\left[ {\begin{gathered}
				\Dec(\PP,\SK_{\ID},\CT_{\ID})=\bf{m}\end{gathered}  
			\left| \begin{gathered}
				\SK_{\ID}\gets\Extract(\PP,\MSK,\ID)\\
				\CT_{\ID}\gets\Enc(\PP,\ID,\bf{m})
			\end{gathered}  \right.} \right]=1.$$
		\item[(2)] For any security parameter $\lambda$, any identities $\ID_i$, $\ID_j$ and any messages $\bf{m}_i, \bf{m}_j$, it holds that:    
		$$\Pr\left[{
			\Test\left( \begin{gathered}
				\td_{\alpha, \ID_i} \\
				\td_{\alpha, \ID_j} \\
				\CT_{\ID_i} \\
				\CT_{\ID_j} \\ 
			\end{gathered}  \right) = 1\left| \begin{array}{l}
				\SK_{\ID_i}\gets\Extract(\PP,\MSK,\ID_i) \\
				\CT_{\ID_i}\gets\Enc(\PP,\ID_i,\bf{m}_i) \\
				\td_{\alpha, \ID_i}\gets\Td_\alpha () \\
				\SK_{\ID_j}\gets\Extract(\PP,\MSK,\ID_j) \\
				\CT_{\ID_j}\gets\Enc(\PP,\ID_j,\bf{m}_j) \\
				\td_{\alpha, \ID_j}\gets\Td_\alpha () 
			\end{array}  \right.} \right]=1.$$
		if $\bf{m}_i=\bf{m}_j$, regardless $i$ and $j$ value, where $\alpha = 1, 2, 3$.
		\item[(3)] For any security parameter $\lambda$, any identities $\ID_i$, $\ID_j$ and any messages $\bf{m}_i, \bf{m}_j$, it holds that:    
		$$\Pr\left[{
			\Test\left( \begin{gathered}
				\td_{\alpha, \ID_i} \\
				\td_{\alpha, \ID_j} \\
				\CT_{\ID_i} \\
				\CT_{\ID_j} \\ 
			\end{gathered}  \right) = 1\left| \begin{array}{l}
				\SK_{\ID_i}\gets\Extract(\PP,\MSK,\ID_i) \\
				\CT_{\ID_i}\gets\Enc(\PP,\ID_i,\bf{m}_i) \\
				\td_{\alpha, \ID_i}\gets\Td_\alpha () \\
				\SK_{\ID_j}\gets\Extract(\PP,\MSK,\ID_j) \\
				\CT_{\ID_j}\gets\Enc(\PP,\ID_j,\bf{m}_j) \\
				\td_{\alpha, \ID_j}\gets\Td_\alpha () 
			\end{array}  \right.} \right]=negligible(\lambda)$$
		if for any ciphertexts $\CT_{\ID_i}$, $\CT_{\ID_j}$ and $\Dec(\PP,\SK_{\ID_i},\CT_{\ID_i})\ne\Dec(\PP,\SK_{\ID_j},\CT_{\ID_j})$, regardless of whether $i=j$, where $\alpha = 1, 2, 3$.
	\end{description}
	
	\noindent\textbf{Security model of IBEET-FA.}
	For the IBEET-FA security model, we consider two types of adversaries:
	\begin{itemize}
		\item[$\bullet$] Type-I adversary: the adversaries can perform authorization (equality) tests on the challenge ciphertext by requesting to obtain a trapdoor for authorization of the target identity. Their goal is to reveal the plain text message in the challenge ciphertext.
		\item[$\bullet$] Type-II adversary: the adversaries cannot perform authorization (equality) tests on the challenge ciphertext since they cannot obtain a trapdoor for authorization of the target identity. Their goal is to distinguish which message is in the challenge ciphertext between two candidates.
	\end{itemize}
	
	Note that the Type-3 authorization is actually a combination of Type-1 and Type-2 authorizations. Hence the Type-3 authorization queries are intentionally omitted. For simplicity, we only provide Type-$\alpha$ $(\alpha = 1, 2)$ authorization queries to the adversary in security games. The security model of a IBEET-FA scheme against two types of adversaries above is described as follows.\\
	
	\noindent\textbf{OW-ID-CPA security against Type-I adversaries.}
	We present the game between a challenger $\cal{C}$ and a Type-I adversary $\cal{A}$. Let say $\cal{A}$ wants to attack target identity $\ID^*$ and he can have a trapdoor for all ciphertexts of that identity.
	\begin{enumerate}
		\item \textbf{Setup:} $\cal{C}$ runs $\Setup(\lambda)$ to generate the pair $(\PP,\MSK)$, and sends the public parameter $\PP$ to $\cal{A}$.
		\item \textbf{Phase 1:} $\cal{A}$ may adaptively make queries many times and in any order to the following oracles:
		\begin{itemize}
			\item $\cal{O}^{\Ext}$: on input identity $\ID$ ($\ne \ID^*$), $\cal{O}^{\Ext}$ returns the $\ID$'s secret key $\SK_{\ID}$.
			\item $\cal{O}^{\Td_\alpha}$ for $\alpha = 1, 2$:
			\begin{itemize}
				\item[$\bullet$] oracle $\cal{O}^{\Td_1}$ return $\td_{1, \ID}\gets\Td_1(\SK_{\ID})$ using the secret key $\SK_{\ID}$.
				\item[$\bullet$] oracle $\cal{O}^{\Td_2}$ return $\td_{2, \ID}\gets\Td_2(\SK_{\ID}, \CT_{\ID})$ using the secret key $\SK_{\ID}$ and $\CT_\ID$.
			\end{itemize}
		\end{itemize}
		\item \textbf{Challenge:} $\cal{C}$ chooses a random message $\bf{m}$ in the message space and sends $\CT^*_{\ID^*}\gets\Enc(\PP,\ID^*,\bf{m})$ to $\cal{A}$.
		\item \textbf{Phase 2:} same as in Phase $1$ except that $\cal{A}$ cannot be queried to the secret key extraction oracle $\cal{O}^{\Ext}$ for the identity $\ID^*$.
		\item \textbf{Guess:} $\cal{A}$ output $\bf{m}'$.
	\end{enumerate}
	The adversary $\cal{A}$ wins the above game if $\bf{m}=\bf{m}'$ and the success probability of $\cal{A}$ is defined as
	$$\Adv^{\OW, Type-\alpha}_{\cal{A},\text{IBEET-FA}}(\lambda):=\Pr[\bf{m}=\bf{m}'].$$
	
	\begin{rem}
		It is necessary to have the size of the message space $\cal{M}$ to be exponential in the security parameter and the min-entropy of the message distribution is sufficiently higher than the security parameter. If not, a Type-I adversary $\cal{A}$ can perform the equality tests with the challenge ciphertext and all other ciphertexts from messages space generated the adversary himself by using the trapdoor for the challenge ciphertext. Then he can reveal the underlying message of the challenge ciphertext in polynomial-time or sufficiently small exponential time in the security parameter. For instance, an algorithm is said to be exponential time, if $T(n)$ is upper bounded by $2^{poly(n)}$, where $poly(n)$ is some polynomial in n. More informally, an algorithm is exponential time if $T(n)$ is bounded by $O(2^{n^k})$. If both $n$ and $k$ are small, then the total time is still sufficiently acceptable.
	\end{rem}
	
	\noindent\textbf{IND-ID-CPA security against Type-II adversaries.}
	We illustrate the game between a challenger $\cal{C}$ and a Type-II adversary $\cal{A}$, who cannot have a trapdoor for all ciphertexts of target identity $\ID^*$.
	\begin{enumerate}
		\item \textbf{Setup:} $\cal{C}$ runs $\Setup(\lambda)$ to generate $(\PP,\MSK)$ and gives the public parameter $\PP$ to $\cal{A}$.
		\item \textbf{Phase 1:} $\cal{A}$ may adaptively make queries many times and in any order to the following oracles:
		\begin{itemize}
			\item $\cal{O}^{\Ext}$: on input identity $\ID$ ($\ne \ID^*$), returns the $\ID$'s secret key $\SK_{\ID}$.
			\item $\cal{O}^{\Td_\alpha}$ for $\alpha = 1, 2$ (where input $\ID \ne \ID^*$):
			\begin{itemize}
				\item[$\bullet$] oracle $\cal{O}^{\Td_1}$: on input $\ID$, $\cal{O}^{\Td_1}$ returns $\td_{1, \ID}\gets\Td_1(\SK_{\ID})$.
				\item[$\bullet$] oracle $\cal{O}^{\Td_2}$: on input ($\ID, \CT_{\ID}$), $\cal{O}^{\Td_2}$ returns $\td_{2, \ID}\gets\Td_2(\SK_{\ID}, \CT_{\ID})$.
			\end{itemize}
		\end{itemize}
		
		\item \textbf{Challenge:} $\cal{A}$ selects a target user $\ID^*$, which was never queried to the $\cal{O}^{\Ext}$ and  $\cal{O}^\Td$ oracles in Phase 1, and two messages $\bf{m}_0$ $\bf{m}_1$ of same length and pass to $\cal{C}$, who then selects a random bit $b\in\{0,1\}$, runs $\CT^*_{\ID^*, b}\gets\Enc(\PP,\ID^*,\bf{m}_b)$ and sends $\CT^*_{\ID^*,b}$ to $\cal{A}$.
		
		
		\item \textbf{Phase 2:} same as in Phase $1$ except that $\cal{A}$ cannot query $\cal{O}^{\Ext}$ and $\cal{O}^{\Td_\alpha}$ for the identity $\ID^*$.
		\item \textbf{Guess:} $\cal{A}$ output $b'$.
	\end{enumerate}
	The adversary $\cal{A}$ wins the above game if $b=b'$ and the advantage of $\cal{A}$ is defined as
	$$\Adv_{\cal{A},\text{IBEET-FA}}^{\IND, Type-\alpha}:=\left|\Pr[b=b']-\frac{1}{2}\right|.$$

	\subsection{Lattices}
	We mainly focus on integer lattices, namely discrete subgroups of $\ZZ^m$. Specially, a lattice $\Lambda$ in $\ZZ^m$ with basis $B=[\b_1,\cdots,\b_n]\in\ZZ^{m\times n}$, where each $\b_i$ is written in column form, is defined as
	$$\Lambda:=\left\{\sum_{i=1}^n\b_ix_i | x_i\in\ZZ~\forall i=1,\cdots,n \right\}\subseteq\ZZ^m.$$
	We call $n$ the rank of $\L$ and if $n=m$ we say that $\L$ is a full rank lattice. In this paper, we mainly consider full rank lattices containing $q\ZZ^m$, called $q$-ary lattices, defined as the following, for a given matrix $A\in\ZZ^{n\times m}$ and $\bf{u}\in\ZZ_q^n$
	\begin{align*}
		\L_q(A) &:= \left\{ \e\in\ZZ^m ~\rm{s.t.}~ \exists \bf{s}\in\ZZ_q^n~\rm{where}~A^T\bf{s}=\bf{e}\mod q \right\}\\
		\Lp_q(A) &:= \left\{ \e\in\ZZ^m~\rm{s.t.}~A\e=0\mod q \right\} \\
		\L_q^{\bf{u}}(A) &:=  \left\{ \e\in\ZZ^m~\rm{s.t.}~A\e=\bf{u}\mod q \right\}
	\end{align*}
	Note that if $\bf{t}\in\L_q^{\bf{u}}(A)$ then $\L_q^{\bf{u}}(A)=\Lp_q(A)+\bf{t}$. Hence, one can see $\L_q^{\bf{u}}(A)$ as a shift of $\Lp_q(A)$.
	
	Let $S=\{\s_1,\cdots,\s_k\}$ be a set of vectors in $\mathbb{R}^m$. We denote by $\|S\|:=\max_i\|\s_i\|$ for $i=1,\cdots,k$, the maximum $l_2$ length of the vectors in $S$. We also denote $\tilde{S}:=\{\tilde{\s}_1,\cdots,\tilde{\s}_k \}$ the Gram-Schmidt orthogonalization of the vectors $\s_1,\cdots,\s_k$ in that order. We refer to $\|\tilde{S}\|$ the Gram-Schmidt norm of $S$.
	
	Ajtai~\cite{Ajtai99} first proposed how to sample a uniform matrix $A\in\ZZ_q^{n\times m}$ with an associated basis $S_A$ of $\Lp_q(A)$ with low Gram-Schmidt norm. It is improved later by Alwen and Peikert~\cite{AP09} in the following Theorem.
	
	\begin{theorem}\label{thm:TrapGen}
		Let $q\geq 3$ be odd and $m:=\lceil 6n\log q\rceil$. There is a probabilistic polynomial-time algorithm $\TrapGen(q,n)$ that outputs a pair $(A\in\ZZ_q^{n\times m},S\in\ZZ^{m\times m})$ such that $A$ is statistically close to a uniform matrix in $\ZZ_q^{n\times m}$ and $S$ is a basis for $\Lp_q(A)$ satisfying
		\[\|\tilde{S}\|\leq O(\sqrt{n\log q})\quad\text{and}\quad\|S\|\leq O(n\log q)\]
		with all but negligible probability in $n$.
	\end{theorem}
	
	\begin{definition}[Gaussian distribution]
		Let $\L\subseteq\ZZ^m$ be a lattice. For a vector $\bf{c}\in\RR^m$ and a positive parameter $\sigma\in\RR$, define:
		$$\rho_{\sigma,\c}(\x)=\exp\left(\pi\frac{\|\x-\c\|^2}{\sigma^2}\right)\quad\text{and}\quad
		\rho_{\sigma,\c}(\L)=\sum_{\x\in\L}\rho_{\sigma,\c}(\x).    $$
		The discrete Gaussian distribution over $\L$ with center $\c$ and parameter $\sigma$ is
		$$\forall \bf{y}\in\L\quad,\quad\cal{D}_{\L,\sigma,\c}(\bf{y})=\frac{\rho_{\sigma,\c}(\bf{y})}{\rho_{\sigma,\c}(\L)}.$$
	\end{definition}
	For convenience, we will denote by $\rho_\sigma$ and $\cal{D}_{\L.\sigma}$ for $\rho_{\bf{0},\sigma}$ and $\cal{D}_{\L,\sigma,\bf{0}}$ respectively. When $\sigma=1$ we will write $\rho$ instead of $\rho_1$. 
	
	We recall below in Theorem~\ref{thm:Gauss} some useful results. The first one comes from~{\cite[Lemma 4.4]{MR04}} . The second one is from~\cite{CHKP10} and formulated in ~{\cite[Theorem 17]{ABB10-EuroCrypt}} and the last one is from~{\cite[Theorem 19]{ABB10-EuroCrypt}}.
	
	\begin{theorem}\label{thm:Gauss}
		Let $q> 2$ and let $A, B$ be matrices in $\ZZ_q^{n\times m}$ with $m>n$ and $B$ is rank $n$. Let $T_A, T_B$ be a basis for $\Lp_q(A)$ and  $\Lp_q(B)$ respectively.
		Then for $c\in\RR^m$ and $U\in\ZZ_q^{n\times t}$:
		\begin{enumerate}
			\item Let $M$ be a matrix in $\ZZ_q^{n\times m_1}$ and $\sigma\geq\|\widetilde{T_A}\|\omega(\sqrt{\log(m+m_1)})$. Then there exists a PPT algorithm $\SampleLeft(A,M,T_A,U,\sigma)$ that outputs a vector $\e\in\ZZ^{m+m_1}$ distributed statistically close to $\cal{D}_{\L_q^{\u}(F_1),\sigma}$ where $F_1:=(A~|~M)$. In particular $\e\in \L_q^{U}(F_1)$, i.e., $F_1\cdot\e=U\mod q$.
			In addition, if $A$ is rank n then there is a PPT algorithm $\SampleBasisLeft(A, M, T_A, \sigma)$ that outputs a basis of $\Lambda_q^\perp(F_1)$.\\
			
			\item Let $R$ be a matrix in $\ZZ^{k\times m}$ and let $s_R:=\sup_{\|\x\|=1}\|R\x\|$. Let $F_2:=(A~|~AR+B)$. Then for  $\sigma\geq\|\widetilde{T_B}\|s_R\omega(\sqrt{\log m})$, there exists a PPT algorithm \\$\SampleRight(A,B,R,T_B,U,\sigma)$ that outputs a vector $\e\in\ZZ^{m+k}$ distributed statistically close to $\cal{D}_{\L_q^{U}(F_2),\sigma}$. In particular $\e\in \L_q^{\u}(F_2)$, i.e., $F_2\cdot\e=U\mod q$.
			In addition, if $B$ is rank n then there is a PPT algorithm\\ $\SampleBasisRight(A, B, R, T_B,\sigma)$ that outputs a basis of $\Lambda_q^\perp(F_2)$.\\
			Note that when $R$ is a random matrix in $\{-1,1\}^{m\times m}$ then $s_R<O(\sqrt{m})$ with overwhelming probability (cf.~{\cite[Lemma 15]{ABB10-EuroCrypt}}).
		\end{enumerate}
	\end{theorem}
	
	The security of our construction reduces to the LWE (Learning With Errors) problem introduced by Regev~\cite{Regev05}.
	\begin{definition}[LWE problem]
		Consider publicly a prime $q$, a positive integer $n$, and a distribution $\chi$ over $\ZZ_q$. An $(\ZZ_q,n,\chi)$-LWE problem instance consists of access to an unspecified challenge oracle $\cal{O}$, being either a noisy pseudorandom sampler $\cal{O}_\s$ associated with a secret $\s\in\ZZ_q^n$, or a truly random sampler $\cal{O}_\$$ who behaviors are as follows:
		\begin{description}
			\item[$\cal{O}_\s$:] samples of the form $(\u_i,v_i)=(\u_i,\u_i^T\s+x_i)\in\ZZ_q^n\times\ZZ_q$ where $\s\in\ZZ_q^n$ is a uniform secret key, $\u_i\in\ZZ_q^n$ is uniform and $x_i\in\ZZ_q$ is a noise withdrawn from $\chi$.
			\item[$\cal{O}_\$$:] samples are uniform pairs in $\ZZ_q^n\times\ZZ_q$.
		\end{description}\
		The $(\ZZ_q,n,\chi)$-LWE problem allows responds queries to the challenge oracle $\cal{O}$. We say that an algorithm $\cal{A}$ decides the $(\ZZ_q,n,\chi)$-LWE problem if 
		$$\Adv_{\cal{A}}^{\mathsf{LWE}}:=\left|\Pr[\cal{A}^{\cal{O}_\s}=1] - \Pr[\cal{A}^{\cal{O}_\$}=1] \right|$$    
		is non-negligible for a random $\s\in\ZZ_q^n$.
	\end{definition}
	Regev~\cite{Regev05} showed that (see Theorem~\ref{thm:LWE} below) when $\chi$ is the distribution $\overline{\Psi}_\alpha$ of the random variable $\lfloor qX\rceil\mod q$ where $\alpha\in(0,1)$ and $X$ is a normal random variable with mean $0$ and standard deviation $\alpha/\sqrt{2\pi}$ then the LWE problem is hard. 
	\begin{theorem}\label{thm:LWE}
		If there exists an efficient, possibly quantum, algorithm for deciding the $(\ZZ_q,n,\overline{\Psi}_\alpha)$-LWE problem for $q>2\sqrt{n}/\alpha$ then there is an efficient quantum algorithm for approximating the SIVP and GapSVP problems, to within $\tilde{\cal{O}}(n/\alpha)$ factors in the $l_2$ norm, in the worst case.
	\end{theorem}
	Hence if we assume  the hardness of approximating the SIVP and GapSVP problems in lattices of dimension $n$ to within polynomial (in $n$) factors, then it follows from Theorem~\ref{thm:LWE} that deciding the LWE problem is hard when $n/\alpha$ is a polynomial in $n$.
	
	
	\section{Proposed Construction: IBEET-FA over Integer Lattice}

	\subsection{Construction}
	\begin{description}
		\item[Setup($\lambda$)]: On input a security parameter $\lambda$, set the parameters $q,n,m,\sigma,\alpha$ as in section \ref{sec:params}
		\begin{enumerate}
			\item Use $\TrapGen(q,n)$ to generate uniformly random $n\times m$-matrices $A, A'\in\ZZ_q^{n\times m}$ together with trapdoors $T_{A}$, $T_{A'}\in\ZZ_q^{m\times m}$ respectively.
			\item Select $\ell+1$ uniformly random $n\times m$ matrices $A_1,\cdots,A_\ell,B\in\ZZ_q^{n\times m}$ where $\ell$ is the bit length of identity vector $\ID$.
			\item Select a uniformly random matrix $U\in\ZZ_q^{n\times t}$ where $t$ is the bit length of message $\m$.
			\item $H: \{0,1\}^*\to \{0,1\}^t$ is a collision-resistant hash function.
			\item $H': \{0,1\}^*\to \{0,1\}^\lambda$ is a collision-resistant hash function.
			\item Output the public key and the master secret key
			$$\PP=(A,A',A_1,\cdots,A_\ell,B,U)\quad,\quad \MSK=(T_A,T_{A'}).$$
		\end{enumerate} 
		
		\item[Extract($\PP,\MSK,\ID$)]: On input the public parameter $\PP$, a master secret key $\MSK$ and an identity $\ID=(\id_1,\cdots,\id_\ell)\in\{-1,1\}^\ell$:
		
		\begin{enumerate}
			\item Let $A_{\ID} = B + \sum_{i=1}^\ell\id_iA_i\in\ZZ_q^{n\times m}$.
			\item Set $F_\ID=(A|A_\ID), F'_\ID=(A'|A_\ID)\in\ZZ_q^{n\times 2m}$.
			\item Sample $E_{\ID}, E'_{\ID}\in\ZZ^{2m\times 2m}$ as 
			$$E_{\ID}\gets\SampleBasisLeft(A,A_{\ID},T_A,\sigma)$$
			$$E'_{\ID}\gets\SampleBasisLeft(A',A_{\ID},T_{A'},\sigma)$$
			\item Output $\SK_\ID:=(E_{\ID},E'_{\ID})$.
		\end{enumerate}
		Note that $F_\ID\cdot E_{\ID} = 0, F'_\ID\cdot E'_{\ID} = 0$ in $\ZZ_q$ and $E_{\ID},E'_{\ID}$ are basis of $\Lambda_q^\perp(F_\ID)$, $\Lambda_q^\perp(F'_\ID)$ respectively.\\
		
		\item[Encrypt($\PP,\ID,\bf{m}$)]: On input the public parameter $\PP$, an identity $\ID=(\id_1,\cdots,\id_\ell)\in\{-1,1\}^\ell$ and a message $\bf{m}\in\{0,1\}^t$, do:
		\begin{enumerate}
			\item Choose uniformly random $\bf{s}_1, \bf{s}_2\in\ZZ_q^n$
			\item Choose $\bf{x}_1,\bf{x}_2\in\overline{\Psi}_\alpha^t$ and compute
			$$\c_1 = U^T\bf{s}_1 +\bf{x}_1 +\bf{m}\big\lfloor\frac{q}{2}\big\rfloor \in\ZZ_q^t$$
			$$\c_2 = U^T\bf{s}_2 +\bf{x}_2 +H(\bf{m})\big\lfloor\frac{q}{2}\big\rfloor \in\ZZ_q^t$$
			\item Pick $\ell$ uniformly random matrices $R_i\in\{-1,1\}^{m\times m}$ for $i=1,\cdots,\ell$, set
			$$R_{\ID}=\sum_{i=1}^\ell \id_iR_i\in\{-\ell,\cdots,\ell\}^{m\times m}$$
			\item Choose a uniformly at random $R\in\{-\ell,\cdots, \ell\}^{m\times m}$ and set
			$$F_1=(F_\ID|AR)\quad,\quad F_2=(F'_\ID|AR)\in\ZZ_q^{n\times 3m}$$
			\item Choose $\bf{y}_1, \bf{y}_2\in\overline{\Psi}_\alpha^m$ and set
			$$\bf{z}_1=R_{\ID}^T\bf{y}_1 \quad,\quad \bf{z}_2=R_{\ID}^T\bf{y}_2\in\ZZ_q^{m}$$
			$$\bf{r}_1=R^T\bf{y}_1 \quad,\quad \bf{r}_2=R^T\bf{y}_2\in\ZZ_q^{m}$$			
			\item Compute
			$$\c_3=F_1^T\bf{s}_1+ \left[\bf{y}_1 | \bf{z}_1 | \bf{r}_1 \right]^T\in\ZZ_q^{3m}$$
			$$\c_4=F_2^T\bf{s}_2+ \left[\bf{y}_2 | \bf{z}_2 | \bf{r}_2 \right]^T\in\ZZ_q^{3m}$$
			\item Compute $\c_5=H'(R||\c_1||\c_2||\c_3||\c_4)$.
			\item The ciphertext is
			$$\CT_\ID=(R,\c_1,\c_2,\c_3,\c_4,\c_5)\in\ZZ_q^{m \times m} \times \ZZ_q^{2t+6m} \times \{0,1\}^\lambda$$    
		\end{enumerate}
		
		\item[Decrypt($\PP,\SK_\ID,\CT$)]: On input public parameter $\PP$, private key $\SK_\ID=(E_\ID,E'_\ID)$ and a ciphertext $\CT=(R,\c_1,\c_2,\c_3,\c_4,\c_5)$, do:
		\begin{enumerate}
			\item Check if $\c_5=H'(R||\c_1||\c_2||\c_3||\c_4)$. If not, returns $\bot$.
			\item Sample $\e\in\ZZ_q^{3m\times t}$ as 
			$$\e\gets\SampleLeft(F_\ID,AR,E_\ID,U,\sigma)$$
			\item Note that $F_1\cdot\e=U\in\ZZ_q^{n\times t}$. We compute $\bf{w}\gets\c_1-\e^T\c_3\in\ZZ_q^t$.
			\item For each $i=1,\cdots, t$, compare $w_i$ and $\lfloor\frac{q}{2}\rfloor$. If they are close, i.e. $\left|w_i - \lfloor\frac{q}{2}\rfloor \right| < \lfloor\frac{q}{4}\rfloor$, output $m_i=1$, otherwise $m_i=0$. We then obtain the message $\bf{m}$.
			\item Sample $\e'\in\ZZ_q^{3m\times t}$ as 
			$$\e'\gets\SampleLeft(F'_\ID,AR,E'_\ID,U,\sigma)$$
			\item Note that $F_2\cdot\e'=U\in\ZZ_q^{n\times t}$. We compute $\bf{w}'\gets\c_2-(\e')^T\c_4\in\ZZ_q^t$.
			\item For each $i=1,\cdots,t$, compare $w'_i$ and $\lfloor\frac{q}{2}\rfloor$. If they are close, i.e. $\left|w_i - \lfloor\frac{q}{2}\rfloor \right| < \lfloor\frac{q}{4}\rfloor$, output $h_i=1$, otherwise $h_i=0$. We then obtain the vector $\bf{h}$.
			\item Finally, if $\bf{h}=H(\bf{m})$ then output $\bf{m}$, otherwise output $\perp$.\\
		\end{enumerate}
		
		\item[Flexible Authorization:] Suppose we have two identities in the system $\ID_i$ and $\ID_j$ along with $\CT_{\ID_i} = (R_{i}, \c_{i1}, \c_{i2}, \c_{i3}, \c_{i4},\c_{i5})$ and $\CT_{\ID_j} = (R_{j}, \c_{j1}, \c_{j2}, \c_{j3}, \c_{j4},\c_{j5})$ be their ciphertexts respectively.
		
		\item[Type-1 Authorization]:
		\begin{itemize}
			\item $\textbf{Td}_1(\SK_\ID)$: On input $\SK_\ID=(E_\ID,E'_\ID)$, outputs a trapdoor $\td_{1,\ID}=E'_\ID$.
			\item \textbf{Test}($\td_{1,\ID_i},\td_{1,\ID_j},\CT_{\ID_i},\CT_{\ID_j}$): On input trapdoors $\td_{1,\ID_i}, \td_{1,\ID_j}$ and ciphertexts $\CT_{\ID_i},\CT_{\ID_j}$ for identities $\ID_i, \ID_j$ respectively, do 
			\begin{enumerate}
				\item For each identity $\ID_i$ (resp. $\ID_j$), do the following:
				\begin{itemize}
					\item[$\bullet$] Check if $\c_{i5}=H'(R_i||\c_{i1}||\c_{i2}||\c_{i3}||\c_{i4})$. If not, returns $\bot$.
					\item[$\bullet$] Take $R_i$ from $\CT_{\ID_i}$ and sample $\e'_i\in\ZZ_q^{3m\times t}$ as
					$$\e'_i\gets\SampleLeft(F'_\ID,AR_i,E'_\ID,U,\sigma)$$
					Note that $F_{i2}\cdot\e'_i=U\in\ZZ_q^{n\times t}$
					\item[$\bullet$] Compute $\bf{w'_i}\gets\c_{i2}-(\e'_i)^T\c_{i4}\in\ZZ_q^t$. For each $k=1,\cdots, t$, compare $w_{ik}$ with $\lfloor\frac{q}{2}\rfloor$ and output $\bf{h}_{ik}=1$ if they are close, and $0$ otherwise. We obtain the vector $\bf{h}_i$ (resp. $\bf{h}_j$).
				\end{itemize}
				\item Output $1$ if  $\bf{h}_i=\bf{h}_j$ and $0$ otherwise.
			\end{enumerate}
		\end{itemize}
		
		\item[Type-2 Authorization]:
		\begin{itemize}
			\item $\textbf{Td}_2(\SK_\ID, \CT_\ID)$: On input $\SK_\ID=(E_\ID,E'_\ID)$ and its ciphertext $\CT_\ID$, outputs a trapdoor $\td_{2,\ID}$ with following procedure:
			\begin{enumerate}
				\item Check if $\c_5=H'(R||\c_{1}||\c_{2}||\c_{3}||\c_{4})$. If not, returns $\bot$.
				\item Take $R$ from $\CT_{\ID}$ and sample $\td_{2,\ID} = \e'\in\ZZ_q^{3m\times t}$ as
				$$\e'\gets\SampleLeft(F'_\ID,AR,E'_\ID,U,\sigma)$$
				Note that $F_2\cdot\e'=U\in\ZZ_q^{n\times t}$
			\end{enumerate}
			\item \textbf{Test}($\td_{2,\ID_i},\td_{2,\ID_j},\CT_{\ID_i},\CT_{\ID_j}$): On input trapdoors $\td_{2,\ID_i}, \td_{2,\ID_j}$ and ciphertexts
			 $\CT_{\ID_i},\CT_{\ID_j}$ for identities $\ID_i, \ID_j$ respectively, do 
			\begin{enumerate}
				\item Compute $\bf{w'_i}\gets\c_{i2}-(\e'_i)^T\c_{i4}\in\ZZ_q^t$. For each $k=1,\cdots, t$, compare $w_{ik}$ with $\lfloor\frac{q}{2}\rfloor$ and output $\bf{h}_{ik}=1$ if they are close, and $0$ otherwise. We then obtain the vector $\bf{h}_i$ (resp. $\bf{h}_j$).
				\item Output $1$ if  $\bf{h}_i=\bf{h}_j$ and $0$ otherwise.
			\end{enumerate}
		\end{itemize}
		
		\item[Type-3 Authorization]:
		\begin{itemize}
			\item $\textbf{Td}_3(\SK_{\ID_i})$: On input $\SK_{\ID_i}=(E_{\ID_i},E'_{\ID_i})$, returns $\td_{3,{\ID_i}}=E'_{\ID_i}$.
			\item $\textbf{Td}_3(\SK_{\ID_j}, \CT_j)$: On input an identity's secret key $\SK_{\ID_j}=(E_{\ID_j},E'_{\ID_j})$ and ciphertext $\CT_j$, it outputs a trapdoor $\td_{3,{\ID_j}}$ using procedure:
			\begin{enumerate}
				\item Check if $\c_{j5}=H'(R_j||\c_{j1}||\c_{j2}||\c_{j3}||\c_{j4})$. If not, returns $\bot$.
				\item Take $R_j$ from $\CT_{\ID_j}$ and sample $\td_{3,{\ID_j}} = \e_j'\in\ZZ_q^{3m\times t}$ as
				$$\e_j'\gets\SampleLeft(F'_\ID,AR_j,E'_\ID,U,\sigma)$$
			\end{enumerate}
			Note that $F_{j2}\cdot\e_j'=U\in\ZZ_q^{n\times t}$.
			
			\item \textbf{Test}($\td_{3,\ID_i},\td_{3,\ID_j},\CT_{\ID_i},\CT_{\ID_j}$): On input trapdoors $\td_{3,\ID_i}, \td_{3,\ID_j}$ and ciphertexts $\CT_{\ID_i},\CT_{\ID_j}$ for identities $\ID_i, \ID_j$ respectively, do 
			\begin{enumerate}
				\item For identity $\ID_i$, do the following:
				\begin{itemize}
					\item[$\bullet$] Check if $\c_{i5}=H'(R_i||\c_{i1}||\c_{i2}||\c_{i3}||\c_{i4})$. If not, returns $\bot$.
					\item[$\bullet$] Take $R_i$ from $\CT_{\ID_i}$ and sample $\e'_i\in\ZZ_q^{3m\times t}$ as
					$$\e'_i\gets\SampleLeft(F'_\ID,AR_i,E'_\ID,U,\sigma)$$
					Note that $F_{i2}\cdot\e'_i=U\in\ZZ_q^{n\times t}$
					\item[$\bullet$] Compute $\bf{w'_i}\gets\c_{i2}-(\e'_i)^T\c_{i4}\in\ZZ_q^t$. For each $k=1,\cdots, t$, compare $w_{ik}$ with $\lfloor\frac{q}{2}\rfloor$ and output $\bf{h}_{ik}=1$ if they are close, and $0$ otherwise. We then obtain the vector $\bf{h}_i$.
				\end{itemize}
				\item For each $\ID_j$: Compute $\bf{w'_j}\gets\c_{j2}-(\e'_j)^T\c_{j4}\in\ZZ_q^t$. For each $k=1,\cdots, t$, compare $w_{ik}$ with $\lfloor\frac{q}{2}\rfloor$ and output $\bf{h}_{jk}=1$ if they are close, and $0$ otherwise. We then obtain the vector $\bf{h}_j$.
				\item Output $1$ if  $\bf{h}_i=\bf{h}_j$ and $0$ otherwise.
			\end{enumerate}
		\end{itemize}
	\end{description}
	
	\begin{theorem}
		Proposed IBEET-FA construction above is correct if $H$  is a collision-resistant hash function.
	\end{theorem}
	\begin{proof}
		We have that if $\CT$ is a valid ciphertext of $\bf{m}$ then the decryption will always output $\bf{m}$. Furthermore, Test of Type-$\alpha$ (for $\alpha = 1, 2, 3$) checks whether $H(\bf{m}_i)=H(\bf{m}_j)$ where $\CT_i$ and $\CT_j$ are valid ciphertexts of two identities $\ID_i$ and $\ID_j$ respectively. It outputs $1$, meaning that $\bf{m}_i=\bf{m}_j$, which is always correct with overwhelming probability since $H$ is collision resistant. Hence, the proposed IBEET-FA described above is correct.
	\end{proof}

	\subsection{Security analysis}
	We show that the our proposed IBEET-FA construction is $\IND$ secure against Type-II adversaries (cf.~Theorem~\ref{thm:IND}) and  $\OW$ secure against Type-I adversaries (cf.~Theorem~\ref{thm:OW}).
	
	\begin{theorem}\label{thm:IND}
		The IBEET-FA system with parameters $(q,n,m,\sigma,\alpha)$ as in Section~\eqref{sec:params}  is $\IND$ secure provided that $H$ is a one-way hash function and the $(\ZZ_q,n,\bar\Psi_\a)$-LWE assumption holds. 
		In particular, suppose there exists a probabilistic algorithm $\cal{A}$ that wins the $\IND$ game with advantage $\epsilon$. Then there is a probabilistic algorithm $\cal{B}$ that solves the $(\ZZ_q,n,\bar\Psi_\a)$-LWE problem. 
	\end{theorem}

	\begin{proof}	
		The proof is similar to that of {\cite[Theorem 25]{ABB10-EuroCrypt}}, which means to show our construction is indistinguishable from random, meaning the challenge ciphertext is indistinguishable from a random element in the ciphertext space. Assume that there is a Type-II adversary $\cal{A}$ who breaks the $\IND$ security of the IBEET-FA scheme with non-negligible probability $\epsilon$. We now construct an algorithm $\cal{B}$ who solves the LWE problem using $\cal{A}$. We denote the adversary $\cal{A}$'s target identity is $\ID^*$ and the challenge ciphertext is $\CT_{\ID^*}^* = (R^*, \c_1^*, \c_2^*, \c_3^*, \c_4^*,\c_5^*)$.
		
		The proof proceeds in a sequence of similar setting games where the first game is identical to the original $\IND$ one.
		
		\begin{description}
			\item[Game 0.] This is the original $\IND$ game between the attacker $\cal{A}$ against the scheme and the $\IND$ challenger. Recall that in this Game 0 the challenger generates public parameter PP by choosing $\ell + 3$ random matrices $A, A', A_1, \cdots, A_\ell, B$ in $\ZZ_q^{n \times m}$. Let $R^*_i\in\{-1,1\}^{m\times m}$ for $i=1,\cdots,\ell$ be the ephemeral random matrices generated when creating the ciphertext $\CT_{\ID^*}^*$.\\
			
			\item[Game 1.] At setup phase, the challenger $\cal{B}$ chooses $\ell$ uniform random matrices $R_i^*$, and $\ell$ random scalars $h_i\in\ZZ_q$ for $i=1,\cdots,\ell$. After that, it generates $A, A', B$ as in Game 0 and constructs the matrices $A_i$ for $i=1,\cdots,\ell$ as
			$$A_i\gets A\cdot R^*_i-h_i\cdot B\in\ZZ_q^{n\times m}.$$
			The remainder of Game 1 is unchanged with $R_i^*$, $i=1,\cdots, \ell$ is used to generate the challenge ciphertext. Note that these $R^*_i\in\{-1,1\}^{m\times m}$ are chosen in advance in setup phase without the knowledge of $\ID^*$. Using~{\cite[Lemma 13]{ABB10-EuroCrypt}}, we prove that the matrices $A_i$ are statistically close to uniform random independent. Similar to the proof of {\cite[Theorem 25]{ABB10-EuroCrypt}}, in the adversary's view, Game 1 is statistically indistinguishable from Game 0.
			
			\item[Game 2.] At the challenge phase, $\cal{B}$ chooses arbitrary message $\bf{m}'$ from the message space and  encrypts $\bf{m}'$ when generating challenge ciphertext. Here we can not expect the behavior of $\cal{A}$. $\cal{A}$ can obtain $H(\bf{m}')$. If $\cal{A}$ outputs $\bf{m}'$, meaning $\cal{A}$ has broken the one-wayness of the hash function $H$. 
			
			
			\item[Game 3.] We add an abort that is independent of adversary's view as follow:
			\begin{itemize}
				\item In the setup phase, the challenger chooses random $h_i\in\ZZ_q$, $i=1,\cdots, \ell$ and keeps these values private.
				\item In the final guess phase, the adversary outputs a guest value $b'\in\{0,1\}$ for $b$. The challenger now does the following:
				\begin{enumerate}
					\item \textbf{Abort check:} for all queries $\SK_\ID$ to the extract secret key oracle $\cal{O}^\Ext$, the challenger checks whether the identity $\ID=(\id_1,\cdots,\id_\ell)$ satisfies $1+\sum_{i=1}^\ell\id_ih_i\ne 0$ and $1+\sum_{i=1}^\ell\id^*_ih_i= 0$. If not then the $\cal{B}$ overwrites $b'$ with a fresh random bit in $\{0,1\}$ and aborts the game. Note that this is unnoticed from the adversary's view and $\cal{B}$ can even abort the game as soon as the condition is true.
					\item \textbf{Artificial abort:} $\cal{B}$ samples a message $\Gamma$ such that $\Pr[\Gamma=1] = \cal{G}($all queries $\SK_\ID)$ where $\cal{G}$ is defined in \cite[Lemma 28]{ABB10-EuroCrypt}. If $\Gamma=1$, $\cal{B}$ overwrites $b'$ with a fresh random bit and make a artificial abort.
				\end{enumerate}    
			\end{itemize}
			
			\item[Game 4.] We choose $A$ is a uniform random matrix in $\ZZ_q^{n\times m}$. However, matrices $A', B$ and their trapdoor $T_{A'}, T_B$ are generated through $\TrapGen(q,n)$, where $T_{A'}, T_B$ are basis of $\Lp_q(A'), \Lp_q(B)$ respectively. The construction of $A_i$ for $i=1,\cdots,\ell$ remains the same: $A_i=AR_i^*-h_iB$. When $\cal{A}$ queries $\cal{O}^{\Ext}(\ID)$ where $\ID=(\id_1,\cdots,\id_\ell)$, the challenger generate secret key using trapdoor $T_B$ as follows:
			\begin{itemize}
				\item Challenger $\cal{B}$ sets
				$$F_\ID:=(A|B+\sum_{i=1}^\ell \id_iA_i) = (A|AR_{\ID} + h_{\ID} B)$$
				where
				\begin{equation}\label{eq:R and h}
					R_\ID\gets\sum_{i=1}^\ell \id_iR_i^*\in\ZZ_q^{m\times m}\quad\text{and}\quad h_\ID\gets 1+\sum_{i=1}^\ell \id_ih_i\in\ZZ_q.
				\end{equation}
				\item If $h_\ID=0$ then abort the game and pretend that the adversary outputs a random bit in $\{0,1\}$ as in Game 3.
				\item Sample 
				$$E_\ID\gets\SampleBasisRight(A,h_\ID B,R_\ID,T_B,\sigma)$$
				Since $h_\ID$ is non-zero, $T_B$ is also a trapdoor for $h_{\ID} B$. Hence the output $E_\ID$ satisfies $F_\ID\cdot E_{\ID} = 0$ in $\ZZ_q$ and $E_{\ID}$ is basis of $\Lambda_q^\perp(F_\ID)$. Moreover, Theorem~\ref{thm:Gauss} shows that  $\sigma>\|\widetilde{T_B}\|s_{R_\ID}\omega(\sqrt{\log m})$ with $s_{R_\ID}:=\|R_\ID\|$
				\item Finally $\cal{B}$ return $\SK=(E_\ID, E'_\ID)$ to $\cal{A}$ where $E'_\ID$ is computed as in construction
				$$E'_{\ID}\gets\SampleBasisLeft(A',A_{\ID},T_{A'},\sigma)$$
			\end{itemize}
			
			The rest of Game 4 is similar to Game 3. In particular, $\cal{B}$ uses abort check in challenge phase and artificial abort in guess phase. Then Game 4 and Game 3 are identical in the adversary's view.
			
			\item[Game 5.] Game 5 is identical to Game 4, except that the challenge ciphertext is always chosen as a random independent element. And thus $\cal{A}$'s advantage is always $0$. The remaining part is to show Game 4 and Game 5 are computationally indistinguishable. If the abort event happens then the games are clearly indistinguishable. Hence, we only only focus on sequences of queries that do not cause an abort.\\
			
			\item[Reduction from LWE.] Recall that an LWE problem instance is provided as a sampling oracle $\cal{O}$ that can be either truly random $\cal{O}_\$$ or a noisy pseudo-random $\cal{O}_s$ for some secret random $s\in\ZZ_q^n$. Suppose now $\cal{A}$ has a non-negligible advantage in distinguishing Game 4 and Game 5, we use $\cal{A}$ to construct $\cal{B}$ to solve the LWE problem as follows.
			
			\item[Instance.] First of all, $\cal{B}$ requests from $\cal{O}$ and receives, for each $j=1,\cdots, t$ a fresh pair of $(\bf{u}_i,d_i)\in\ZZ_q^n\times\ZZ_q$ and for each $i=1,\cdots,m$, a fresh pair of $(\bf{a}_i,v_i)\in\ZZ_q^n\times\ZZ_q$.
			\item[Setup.] $\cal{B}$ constructs the public parameter $\PP$ as follows:
			
			\begin{enumerate}
				\item Assemble the random matrix $A\in\ZZ_q^{n\times m}$ from $m$ of previously given LWE samples by letting the $i$-th column of $A$ to be the $n$-vector $\bf{a}_i$ for all $i=1,\cdots,m$.
				\item Assemble the first $t$ unused  LWE samples $\bf{u}_1,\cdots,\bf{u}_t$ to become a public random matrix $U\in\ZZ_q^{n\times t}$.
				\item Matrix $A', B$ and trapdoor $T_{A'}, T_B$ are generated as in Game 4 while matrices $A_i$ for $i=1,\cdots,\ell$ are constructed as in Game 1.
				\item Set $\PP:=(A,A',A_1,\cdots,A_\ell,B,U)$ and send to $\cal{A}$.
			\end{enumerate}
			\item[Queries.] $\cal{B}$ answers the queries as in Game 4, including aborting the game.    
			\item[Challenge.]    
			Now when $\cal{A}$ sends $\cal{B}$ two messages $\bf{m}_0$ and $\bf{m}_1$ and a target identity $\ID^*=(\id^*_1,\cdots,\id^*_\ell)$. $\cal{B}$ choose a random bit $b\in\{0,1\}$ and computes the challenge ciphertext  $\CT_{\ID^*}^* = (R^*, \c_1^*, \c_2^*, \c_3^*, \c_4^*,\c_5^*)$. for $\bf{m}_b$ as follows:
			\begin{enumerate}
				\item Assemble $\d^*=[d_1,\cdots,d_t]^T\in\ZZ_q^t$, set
				$$\c_1^*\gets\d^*+\bf{m}_b\lfloor\frac{q}{2}\rfloor\in\ZZ_q^t$$
				\item Choose uniformly random $\s_2\in\ZZ_q^n$ and $\x_2\gets\overline{\Psi}_\alpha^t$, compute
				$$\c_2^* = U^T\bf{s}_2 +\bf{x}_2 +H(\bf{m}_b)\big\lfloor\frac{q}{2}\big\rfloor \in\ZZ_q^t$$
				\item Compute $R^*_{\ID^*}:=\sum_{i=1}^\ell \id_i^*R_i^*\in\{-\ell,\cdots,\ell\}^{m\times m}$.
				\item Choose uniformly at random $R^*\in\{-\ell,\cdots, \ell\}^{m\times m}$.
				\item Assemble $\v^*=[v_1,\cdots,v_m]^T\in\ZZ_q^m$. Set
				$$\c_3^*:=\left[
				\begin{array}{c}
					\v^* \\
					(R^*_{\ID^*})^T \v^*\\
					(R^*)^T \v^*
				\end{array}
				\right] \in\ZZ_q^{3m}$$
				\item Choose $\bf{y}_2\gets\overline{\Psi}_\alpha^m$ and set
				$$\c_4^*:=\left[
				\begin{array}{c}
				A'^T\bf{s}_2+\bf{y}_2 \\
				(AR^*_{\ID^*})^T\bf{s}_2+(R^*_{\ID^*})^T\bf{y}_2\\
				(AR^*)^T\bf{s}_2+(R^*)^T\bf{y}_2
				\end{array}\right]\in\ZZ_q^{3m}$$
				\item Compute $\c_5^*=H'(R^*||\c_1^*||\c_2^*||\c_3^*||\c_4^*)$.
			\end{enumerate}
			Then $\cal{B}$ sends $\CT_{\ID^*}^* = (R^*, \c_1^*, \c_2^*, \c_3^*, \c_4^*, \c_5^*)$ to $\cal{A}$.\\
			
			We argue that when the LWE oracle is pseudorandom (i.e. $\cal{O} = \cal{O}_s$) then $\CT_{\ID^*}^*$ is distributed exactly as in Game 4. It suffices to argue only in case of no abort. Note that, since one has $h_{\ID^*}=0$ and so
			$$F_{\ID^*}=(A|AR^*_{\ID^*})\quad,\quad F^*_1=(A|AR^*_{\ID^*}|AR^*)$$
			Because oracle is pseudorandom, then we have $\v^*=A^T\bf{s}_1+\bf{y}_1$ for some random noise vector $\bf{y}_1\gets\overline{\Psi}_\alpha^m$ respectively. Therefore, $\c_3^*$ in Step 5 satisfies:
			$$\c_3^*:=\left[
			\begin{array}{c}
				A^T\bf{s}_1+\bf{y}_1 \\
				(R^*_{\ID^*})^T (A^T\bf{s}_1+\bf{y}_1)\\
				(R^*)^T (A^T\bf{s}_1+\bf{y}_1)
			\end{array}
			\right]
			=\left[
			\begin{array}{c}
				A^T\bf{s}_1+\bf{y}_1 \\
				(AR^*_{\ID^*})^T\bf{s}_1+(R^*_{\ID^*})^T\bf{y}_1\\
				(AR^*)^T\bf{s}_1 +(R^*)^T\bf{y}_1
			\end{array}\right]
			=(F_1^*)^T\bf{s}_1+\left[
			\begin{array}{c}
				\bf{y}_1\\
				(R^*_{\ID^*})^T\bf{y}_1\\
				(R^*)^T\bf{y}_1
			\end{array}\right]$$
			
			One can easily see that
			$$\c_4^*:=(F_2^*)^T\bf{s}_2+\left[
			\begin{array}{c}
			\bf{y}_2\\
			(R^*_{\ID^*})^T\bf{y}_2\\
			(R^*)^T\bf{y}_2
			\end{array}\right]$$
			
			Moreover, for some $\bf{x}_1,\gets\overline{\Psi}_\alpha^t$, we have $\d^*=U^T\bf{s}_1+\bf{x}_1$, therefore
			$$\c_1^* = U^T\bf{s}_1 +\bf{x}_1 +\bf{m}_b\big\lfloor\frac{q}{2}\big\rfloor \in\ZZ_q^t$$
			Therefore $\CT_{\ID^*}^*$ is a valid ciphertext.
			
			When $\cal{O}=\cal{O}_{\$}$ we have that $\bf{d}^*$ is uniform in $\ZZ_q^t$ and $\v^*$ is uniform in $\ZZ_q^m$. Then obviously $\c_1^*$ is uniform. It follows also from the leftover hash lemma (cf.~{\cite[Theorem 8.38]{Shoup}}) that $\c_3^*$ is also uniform and independence in $\ZZ_q^3m$. Consequently, the challenge ciphertext $\CT_{\ID^*}^*$ is always uniform in $\ZZ_q^{m \times m} \times \ZZ_q^{2t+6m} \times \{0,1\}^\lambda$.\\
			
			\item[Guess.] After making additional queries in Phase 2, $\cal{A}$ guesses if it is interacting with Game 4 or Game 5. The simulator also implements the artificial abort for these Games and uses the final guess as the answer to the LWE problem.
		\end{description}
		
		We claim that when $\cal{O}=\cal{O}_\s$ then the adversary's view is as in Game 4. When $\cal{O}=\cal{O}_\$$ then the view of adversary is as in Game 5. Hence $\cal{B}$'s advantage in solving the LWE problem is the same as the advantage of $\cal{A}$ in distinguishing Game 4 and Game 5 as required. 
		
	\end{proof}
	
	
	
	\begin{theorem}\label{thm:OW}
		The IBEET-FA system with parameters $(q,n,m,\sigma,\alpha)$ as in Section~\eqref{sec:params} is  $\OW$ secure provided that $H$ is a one-way hash function and the $(\ZZ_q,n,\bar\Psi_\a)$-LWE assumption holds.
	\end{theorem}
	
	\begin{proof}
		The proof is similar to that of~\cite[Theorem 25]{ABB10-EuroCrypt}. Assume that there is a Type-I adversary $\cal{A}$ who breaks the $\OW$ security of the IBEET-FA scheme. We construct an algorithm $\cal{B}$ who solves the LWE problem using $\cal{A}$. The proof proceeds in a sequence of games where the first game is identical to the $\OW$.
		
		\begin{description}
			\item[Game 0.] This is the original $\OW$ game between the attacker $\cal{A}$ against the scheme and the $\OW$ challenger. Recall that in this Game 0 the challenger generates public parameter PP by choosing $\ell + 3$ random matrices $A, A', A_1, \cdots, A_\ell, B$ in $\ZZ_q^{n \times m}$. Let $R^*_i\in\{-1,1\}^{m\times m}$ for $i=1,\cdots,\ell$ be the ephemeral random matrices generated when creating the ciphertext $\CT_{\ID^*}^*$.\\
			
			\item[Game 1.] At setup phase, the challenger $\cal{B}$ chooses $\ell$ uniform random matrices $R_i^*$, and $\ell$ random scalars $h_i\in\ZZ_q$ for $i=1,\cdots,\ell$. After that, it generates $A, A', B$ as in Game 0 and constructs the matrices $A_i$ for $i=1,\cdots,\ell$ as
			$$A_i\gets A\cdot R^*_i-h_i\cdot B\in\ZZ_q^{n\times m}.$$
			
			The remainder of Game 1 is unchanged with $R_i^*$, $i=1,\cdots, \ell$ is used to generate the challenge ciphertext. Note that these $R^*_i\in\{-1,1\}^{m\times m}$ are chosen in advance in setup phase without the knowledge of $\ID^*$. Using~{\cite[Lemma 13]{ABB10-EuroCrypt}}, we prove that the matrices $A_i$ are statistically close to uniform random independent. Therefore, the adversary's view in Game 1 is identical to that in Game 0.\\
			
			\item[Game 2.] We add an abort that is independent of adversary's view as follow:
			\begin{itemize}
				\item In the setup phase, the challenger chooses random $h_i\in\ZZ_q$, $i=1,\cdots, \ell$ and keeps these values private.
				\item In the final guess phase, the adversary outputs a guest $\bf{m}'$ for $\bf{m}$. The challenger now does the following:
				\begin{enumerate}
					\item \textbf{Abort check:} for all queries $\SK_\ID$ to the extract secret key oracle $\cal{O}^\Ext$, the challenger checks whether the identity $\ID=(\id_1,\cdots,\id_\ell)$ satisfies $1+\sum_{i=1}^\ell\id_ih_i\ne 0$ and $1+\sum_{i=1}^\ell\id^*_ih_i= 0$. If not then $\cal{B}$ overwrites $\bf{m}'$ with a fresh random message and aborts the game. Note that this is unnoticed from the adversary's view and $\cal{B}$ can even abort the game as soon as the condition is true.
					\item \textbf{Artificial abort:} $\cal{B}$ samples a message $\Gamma$ such that $\Pr[\Gamma=1] = \cal{G}($all queries $\SK_\ID)$ where $\cal{G}$ is defined in \cite[Lemma 28]{ABB10-EuroCrypt}. If $\Gamma=1$, $\cal{B}$ overwrites $\bf{m}'$ with a fresh message and make a artificial abort.
				\end{enumerate}    
			\end{itemize}
			This completes the description of Game 2.\\
			
			\item[Game 3.] We choose $A$ is a uniform random matrix in $\ZZ_q^{n\times m}$. However, matrices $A', B$ and their trapdoor $T_{A'}, T_B$ are generated through $\TrapGen(q,n)$, where $T_{A'}, T_B$ are basis of $\Lp_q(A'), \Lp_q(B)$ respectively. The construction of $A_i$ for $i=1,\cdots,\ell$ remains the same: $A_i=AR_i^*-h_iB$. When $\cal{A}$ queries $\cal{O}^{\Ext}(\ID)$ where $\ID=(\id_1,\cdots,\id_\ell)$, the challenger generate secret key using trapdoor $T_B$ as follows:
			\begin{itemize}
				\item Challenger $\cal{B}$ sets
				$$F_\ID:=(A|B+\sum_{i=1}^\ell \id_iA_i) = (A|AR_{\ID} + h_{\ID} B)$$
				where
				\begin{equation}\label{eq:R and h}
					R_\ID\gets\sum_{i=1}^\ell \id_iR_i^*\in\ZZ_q^{m\times m}\quad\text{and}\quad h_\ID\gets 1+\sum_{i=1}^\ell \id_ih_i\in\ZZ_q.
				\end{equation}
				\item If $h_\ID=0$ then abort the game and pretend that the adversary outputs a random message $\bf{m}'$ as in Game 2.
				\item Sample 
				$$E_\ID\gets\SampleBasisRight(A,h_\ID B,R_\ID,T_B,\sigma)$$
				Since $h_\ID$ is non-zero, $T_B$ is also a trapdoor for $h_{\ID} B$. Hence the output $E_\ID$ satisfies $F_\ID\cdot E_{\ID} = 0$ in $\ZZ_q$ and $E_{\ID}$ is basis of $\Lambda_q^\perp(F_\ID)$. Moreover, Theorem~\ref{thm:Gauss} shows that when $\sigma>\|\widetilde{T_B}\|s_{R_\ID}\omega(\sqrt{\log m})$ with $s_{R_\ID}:=\|R_\ID\|$
				\item Finally $\cal{B}$ return $\SK=(E_\ID, E'_\ID)$ to $\cal{A}$ where $E'_\ID$ is computed as in the real construction
				$$E'_{\ID}\gets\SampleBasisLeft(A',A_{\ID},T_{A'},\sigma)$$
			\end{itemize}
			Note that when $\cal{A}$ queries $\cal{O}^{\Td_\alpha}$ (including $\ID^*$), the trapdoor is still computed as in the real construction.
			
			The rest of Game 3 is similar to Game 2. In particular, $\cal{B}$ uses abort check in challenge phase and artificial abort in guess phase. Then Game 3 and Game 2 are identical in the adversary's view.\\
			
			\item[Game 4.] Game 4 is identical to Game 3, except that the challenge ciphertext is always chosen as a random independent element. And thus $\cal{A}$'s advantage is always $0$. The remaining part is to show Game 3 and Game 4 are computationally indistinguishable. If the abort event happens then the games are clearly indistinguishable. Hence, we only focus on sequences of queries that do not cause an abort.\\
			
			\item[Reduction from LWE.] Recall that an LWE problem instance is provided as a sampling oracle $\cal{O}$ that can be either truly random $\cal{O}_\$$ or a noisy pseudo-random $\cal{O}_s$ for some secret random $s\in\ZZ_q^n$. Suppose now $\cal{A}$ has a non-negligible advantage in distinguishing Game 3 and Game 4, we use $\cal{A}$ to construct $\cal{B}$ to solve the LWE problem as follows.
			
			\item[Instance.] First of all, $\cal{B}$ requests from $\cal{O}$ and receives, for each $j=1,\cdots, t$ a fresh pair of $(\bf{u}_i,d_i)\in\ZZ_q^n\times\ZZ_q$ and for each $i=1,\cdots,m$, a fresh pair of $(\bf{a}_i,v_i)\in\ZZ_q^n\times\ZZ_q$.
			\item[Setup.] $\cal{B}$ constructs the public parameter $\PP$ as follows:
			
			\begin{enumerate}
				\item Assemble the random matrix $A\in\ZZ_q^{n\times m}$ from $m$ of previously given LWE samples by letting the $i$-th column of $A$ to be the $n$-vector $\bf{a}_i$ for all $i=1,\cdots,m$.
				\item Assemble the first $t$ unused  LWE samples $\bf{u}_1,\cdots,\bf{u}_t$ to become a public random matrix $U\in\ZZ_q^{n\times t}$.
				\item Matrix $A', B$ and trapdoor $T_{A'}, T_B$ are generated as in Game 4 while matrices $A_i$ for $i=1,\cdots,\ell$ are constructed as in Game 1.
				\item Set $\PP:=(A,A',A_1,\cdots,A_\ell,B,U)$ and send to $\cal{A}$.
			\end{enumerate}
			\item[Queries.] $\cal{B}$ answers the queries as in Game 4, including aborting the game.    
			\item[Challenge.]    
			Now when $\cal{A}$ send $\cal{B}$ a target identity $\ID^*=(\id^*_1,\cdots,\id^*_\ell)$. $\cal{B}$ then choose a random message $\m^*$ and computes the challenge ciphertext  $\CT_{\ID^*}^* = (R^*, \c_1^*, \c_2^*, \c_3^*, \c_4^*, \c_5^*)$ as follows:
			\begin{enumerate}
				\item Assemble $\d^*=[d_1,\cdots,d_t]^T\in\ZZ_q^t$, set
				$$\c_1^*\gets\d^*+\bf{m}^*\lfloor\frac{q}{2}\rfloor\in\ZZ_q^t$$
				\item Choose uniformly random $\s_2\in\ZZ_q^n$ and $\x_2\gets\overline{\Psi}_\alpha^t$, compute
				$$\c_2^* = U^T\bf{s}_2 +\bf{x}_2 +H(\bf{m}^*)\big\lfloor\frac{q}{2}\big\rfloor \in\ZZ_q^t$$
				\item Compute $R^*_{\ID^*}:=\sum_{i=1}^\ell \id_i^*R_i^*\in\{-\ell,\cdots,\ell\}^{m\times m}$.
				\item Choose uniformly at random $R^*\in\{-\ell,\cdots, \ell\}^{m\times m}$.
				\item Assemble $\v^*=[v_1,\cdots,v_m]^T\in\ZZ_q^m$. Set
				$$\c_3^*:=\left[
				\begin{array}{c}
				\v^* \\
				(R^*_{\ID^*})^T \v^*\\
				(R^*)^T \v^*
				\end{array}
				\right] \in\ZZ_q^{3m}$$
				\item Choose $\bf{y}_2\gets\overline{\Psi}_\alpha^m$ and set
				$$\c_4^*:=\left[
				\begin{array}{c}
				A'^T\bf{s}_2+\bf{y}_2 \\
				(AR^*_{\ID^*})^T\bf{s}_2+(R^*_{\ID^*})^T\bf{y}_2\\
				(AR^*)^T\bf{s}_2+(R^*)^T\bf{y}_2
				\end{array}\right]\in\ZZ_q^{3m}$$
				\item Compute $\c_5^*=H'(R^*||\c_1^*||\c_2^*||\c_3^*||\c_4^*)$.
			\end{enumerate}
			Then $\cal{B}$ sends $\CT_{\ID^*}^* = (R^*, \c_1^*, \c_2^*, \c_3^*, \c_4^*, \c_5^*)$ to $\cal{A}$.\\
			
			We argue that when the LWE oracle is pseudorandom (i.e. $\cal{O} = \cal{O}_s$) then $\CT_{\ID^*}^*$ is distributed exactly as in Game 4. It suffices to argue only in case of no abort. Note that, since one has $h_{\ID^*}=0$ and so
			$$F_{\ID^*}=(A|AR^*_{\ID^*})\quad,\quad F^*_1=(A|AR^*_{\ID^*}|AR^*)$$
			Because oracle is pseudorandom, then we have $\v^*=A^T\bf{s}_1+\bf{y}_1$ for some random noise vector $\bf{y}_1\gets\overline{\Psi}_\alpha^m$ respectively. Therefore, $\c_3^*$ in Step 5 satisfies:
			$$\c_3^*:=\left[
			\begin{array}{c}
			A^T\bf{s}_1+\bf{y}_1 \\
			(R^*_{\ID^*})^T (A^T\bf{s}_1+\bf{y}_1)\\
			(R^*)^T (A^T\bf{s}_1+\bf{y}_1)
			\end{array}
			\right]
			=\left[
			\begin{array}{c}
			A^T\bf{s}_1+\bf{y}_1 \\
			(AR^*_{\ID^*})^T\bf{s}_1+(R^*_{\ID^*})^T\bf{y}_1\\
			(AR^*)^T\bf{s}_1 +(R^*)^T\bf{y}_1
			\end{array}\right]
			=(F_1^*)^T\bf{s}_1+\left[
			\begin{array}{c}
			\bf{y}_1\\
			(R^*_{\ID^*})^T\bf{y}_1\\
			(R^*)^T\bf{y}_1
			\end{array}\right]$$
			
			One can easily see that
			$$\c_4^*:=(F_2^*)^T\bf{s}_2+\left[
			\begin{array}{c}
			\bf{y}_2\\
			(R^*_{\ID^*})^T\bf{y}_2\\
			(R^*)^T\bf{y}_2
			\end{array}\right]$$
			
			Moreover, for some $\bf{x}_1,\gets\overline{\Psi}_\alpha^t$, we have $\d^*=U^T\bf{s}_1+\bf{x}_1$, therefore
			$$\c_1^* = U^T\bf{s}_1 +\bf{x}_1 +\bf{m}^*\big\lfloor\frac{q}{2}\big\rfloor \in\ZZ_q^t$$
			Therefore $\CT_{\ID^*}^*$ is a valid ciphertext.
			
			When $\cal{O}=\cal{O}_{\$}$ we have that $\bf{d}^*$ is uniform in $\ZZ_q^t$ and $\v^*$ is uniform in $\ZZ_q^m$. Then obviously $\c_1^*$ is uniform. It follows also from the leftover hash lemma (cf.~{\cite[Theorem 8.38]{Shoup}}) that $\c_3^*$ is also uniform and independence in $\ZZ_q^3m$. Consequently, the challenge ciphertext $\CT_{\ID^*}^*$ is always uniform in $\ZZ_q^{m \times m} \times \ZZ_q^{2t+6m} \times \{0,1\}^\lambda$.\\
			
			\item[Guess.] After making additional queries in Phase 2, $\cal{A}$ guesses if it is interacting with Game 3 or Game 4. The simulator also implements the artificial abort for these Games and uses the final guess as the answer to the LWE problem.
		\end{description}
		
		We claim that when $\cal{O}=\cal{O}_\s$ then the adversary's view is as in Game 3. When $\cal{O}=\cal{O}_\$$ then the view of adversary is as in Game 4. Hence $\cal{B}$'s advantage in solving the LWE problem is the same as the advantage of $\cal{A}$ in distinguishing Game 3 and Game 4 as required. This completes the description of algorithm $\cal{B}$ and completes the proof.
	\end{proof}
	
	
	\subsection{Setting Parameters}\label{sec:params}
	Following~{\cite[Section 7.3]{ABB10-EuroCrypt}}, we choose our parameters satisfying:
	\begin{itemize}
		\item that the $\TrapGen$ works, i.e., $m>6n\log q$.
		\item that $\sigma$ is large enough for $\SampleLeft$, $\SampleRight$, $\SampleBasisLeft$ and $\SampleBasisRight$ to work, i.e.,
		$\sigma>\max \{ \|\widetilde{T_A}\|\cdot \omega(\sqrt{\log (2m)}), \|\widetilde{T_B}\|\cdot s_{R_\ID}\cdot \omega(\sqrt{\log m})\}$. Note that, when $R$ is a random matrix in $\{-1,1\}^{m\times m}$ then $s_R<O(\sqrt{m})$ with overwhelming probability (cf.~{\cite[Lemma 15]{ABB10-EuroCrypt}}). Hence when $R_\ID$ is a random matrix in $\{-\ell,\ell\}^{m\times m}$ then $s_{R_\ID}<O(\ell\sqrt{m})$. Also, note that $\|\widetilde{T_A}\|, \|\widetilde{T_B}\| \leq O(\sqrt{n\log q})$.
		\item that Regev's reduction for the LWE problem to apply, i.e., $q>2\sqrt{n}/\alpha$.
		\item that our security reduction applies (i.e., $q>2Q$ where $Q$ is the number of identity queries from the adversary).
		
			\item the error term in decryption is less than ${q}/{5}$ with high probability, i.e., $q=\Omega(\sigma m^{3/2})$ and $\alpha<[\sigma lm\omega(\sqrt{\log m})]^{-1}$.
	\end{itemize}

	\section{Conclusion}
	In this paper, we propose a direct construction of IBEET-FA based on the hardness of Learning With Errors problem. This construction inherited the efficiency from direct construction of PKEET by Duong et al~{\cite{PKEET-FA-Duong20,IBEET-Duong19}}. However, in order to support flexible authorization, the ciphertext size and secret key size increase. One can improve it by using elegant methods in construction and security proof. We will leave as a future work for improving our schemes to achieve the CCA2-security as well as reducing the storage size.Additionally, as mentioned earlier, the existing security model for IBEET and hence for IBEET-FA is not sufficiently strong in some scenarios. Therefore, finding  a stronger security model for IBEET and IBEET-FA would also be an interesting future work.

\subsubsection{Acknowledgment.} 
We all thank Ryo Nishimaki and anonymous reviewers for their insightful comments which improve the content and presentation of  this work.  This work is partially supported by the Australian Research Council Linkage Project LP190100984. Huy Quoc Le has been sponsored by a CSIRO Data61 PhD Scholarship and CSIRO Data61 Top-up Scholarship. 
	

\end{document}